\documentclass[11pt]{article}
\usepackage{fullpage}
\usepackage{amsmath}
\usepackage{amsthm}
\usepackage{amsfonts}
\usepackage{amssymb}
\usepackage{enumerate}
\usepackage{xcolor}
\usepackage{bbm}

\usepackage{tikz}
\usetikzlibrary{calc}
\usetikzlibrary{arrows}
\usetikzlibrary{shadows}
\usetikzlibrary{positioning}
\usetikzlibrary{decorations.pathmorphing,shapes}
\usetikzlibrary{shapes,backgrounds}
\usepackage{mathrsfs}

\newcommand{\ignore}[1]{}
\newcommand{\etal}{{\em et al.}~}

\newtheoremstyle{prenum}
  {7pt}
  {7pt}
  {\slshape}
  {0pt}
  {\bfseries}
  { }
  {5pt}
  {\thmnumber{#2}\thmname{ #1}\thmnote{ (#3)}}

\theoremstyle{prenum}

\newtheorem{theorem}{Theorem}[section]
\newtheorem{lemma}[theorem]{Lemma}
\newtheorem{fact}[theorem]{Fact}

\newtheorem{proposition}[theorem]{Proposition}

\newtheorem{definition}{Definition}[section]
\newtheorem{assumption}[definition]{Assumption}

\newcommand{\ZZ}{\mathbb{Z}}
\newcommand{\RR}{\mathbb{R}}
\newcommand{\CC}{\mathbb{C}}
\newcommand{\FF}{\mathbb{F}}

\DeclareMathOperator{\Sp}{Spec}

\DeclareMathOperator{\Tr}{Tr}
\DeclareMathOperator{\Supp}{Supp}
\DeclareMathOperator{\spn}{span}
\DeclareMathOperator{\Adj}{Adj}
\DeclareMathOperator{\mat}{Mat}

\newcommand{\Exp}{\mathbb{E}}

\newcommand{\ket}[1]{{#1}}
\newcommand{\braket}[2]{\langle #1 , #2 \rangle}
\newcommand{\ketbra}[2]{{#1}{#2}^\dagger}

\newcommand{\norm}[1]{\left\lVert #1 \right\rVert}
\newcommand{\ip}[2]{\langle #1, #2 \rangle}

\newcommand{\dt}{\frac{d}{dt}}

\newcommand{\qbinom}[2]{\mbox{$\begin{bmatrix} {#1} \\ {#2} \end{bmatrix}_{q}$}}
\newcommand{\eone}{\epsilon_1}

\title{Of Shadows and Gaps in Spatial Search}
\author{Ada Chan\thanks{Department of Mathematics and Statistics, York University.}
\and Chris Godsil\thanks{Department of Combinatorics and Optimization, University of Waterloo.}
\and Christino Tamon\thanks{Department of Computer Science, Clarkson University. Contact: tino@clarkson.edu.}
\and Weichen Xie\thanks{Department of Mathematics, Clarkson University}}

\date{\today}
\begin{document}
\maketitle
\begin{abstract}
Spatial search occurs in a connected graph if a continuous-time quantum walk on the adjacency matrix
of the graph, suitably scaled, plus a rank-one perturbation induced by any vertex will unitarily map
the principal eigenvector of the graph to the characteristic vector of the vertex.
This phenomenon is a natural continuous-time analogue of Grover search.
The spatial search is said to be optimal if it occurs with constant fidelity and in time inversely proportional to
the shadow of the target vertex on the principal eigenvector.
Extending a result of Chakraborty \etal ({\em Physical Review A}, {\bf 102}:032214, 2020),
we prove a simpler characterization of optimal spatial search.
Based on this characterization, we observe that some families of distance-regular graphs,
such as Hamming and Grassmann graphs, have optimal spatial search.
We also show a matching lower bound on time for spatial search with constant fidelity,
which extends a bound due to Farhi and Gutmann for perfect fidelity.
Our elementary proofs employ standard tools, such as Weyl inequalities and Cauchy determinant formula.

\vspace{.1in}
\par\noindent{\em Keywords}: Quantum walk, spatial search, spectral gap, perturbation.
\end{abstract}

\section{Introduction}

In the seminal work \cite{g97}, Grover described a quantum algorithm with a provable quadratic speedup 
for the ubiquitous search problem. 
It was realized later that his algorithm can be viewed as a discrete-time quantum walk on the
complete graph \cite{a07,s04}. In another fundamental work, Farhi and Gutmann \cite{fg98} proposed a 
continuous-time analog of Grover search. Their work was generalized by Childs and Goldstone \cite{cg04} to arbitrary
graphs where the problem is known as {\em spatial search}.

\begin{figure}[t]
\begin{center}
\begin{tikzpicture}[
    main node/.style={circle,draw,font=\bfseries}, main edge/.style={-,>=stealth'},
    scale=0.5,
    stone/.style={},
    black-stone/.style={black!80},
    black-highlight/.style={outer color=black!80, inner color=black!30},
    black-number/.style={white},
    white-stone/.style={white!70!black},
    white-highlight/.style={outer color=white!70!black, inner color=white},
    white-number/.style={black}]
\tikzset{every loop/.style={thick, min distance=17mm, in=45, out=135}}

\tikzstyle{every node}=[draw, thick, shape=circle, circular drop shadow, fill={white}];
\path (0,+3.5) node (p0) [scale=0.8] {};
\path (-2.85,2.05) node (p6) [scale=0.8] {};
\path (+2.85,2.05) node (p1) [scale=0.8] {};
\path (+2.85,-0.1) node (p2) [scale=0.8] {};
\path (-2.85,-0.1) node (p5) [scale=0.8] {};
\path (+1.5,-2) node (p3) [scale=0.8] {};
\path (-1.5,-2) node (p4) [scale=0.8] {};
\draw[thick]
    (p0) -- (p1) -- (p2) -- (p3) -- (p4) -- (p5) -- (p6) -- (p0);
\draw[thick]
    (p0) -- (p2) -- (p4) -- (p6) -- (p1) -- (p3) -- (p5) -- (p0);
\draw[thick]
	(p0) -- (p3) -- (p6) -- (p2) -- (p5) -- (p1) -- (p4) -- (p0);
\draw[thick]
	(p0) -- (p4) -- (p1) -- (p5) -- (p2) -- (p6) -- (p3) -- (p0);
\path[-]
    (p0) edge [loop above] (p0);

\end{tikzpicture}
\vspace{.1in}
\caption{Grover search in continuous-time (see Farhi and Gutmann \cite{fg98}):
a continuous-time quantum walk with $H = \gamma A(K_n) + e_w e_w^T$ can perfectly transfer 
the density matrix $E_1=\frac{1}{n}J_n$ to the target state $e_w e_w^T$, for some $\gamma > 0$, 
in time $\frac{\pi}{2}\sqrt{n}$.
Here, $e_w$ denote the unit vector associated to vertex $w$ of the clique.
}
\label{fig:clique}
\end{center}
\end{figure}
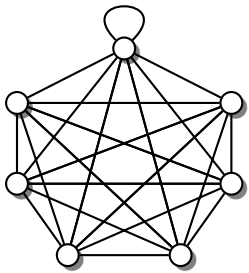

Suppose $G$ is an undirected and connected graph on $n$ vertices with normalized adjacency matrix $A$
whose spectral decomposition is given by $A = \sum_{r=1}^{d} \theta_r E_r$, where $1 = \theta_1 > \theta_2 > \ldots > \theta_d \ge 0$
and $E_r$ is the orthogonal projection onto the eigenspace corresponding to $\theta_r$.
We say $G$ has optimal spatial search if for any vector $w$ (which may correspond to the characteristic vector of a vertex of $G$), 
the continuous-time quantum walk
\[
	\rho(t) = e^{-it H}\rho(0)e^{it H} 
\]
with the time-independent Hamiltonian $H = \gamma A + \ketbra{w}{w}$, 
for a scaling factor $\gamma > 0$, maps the density matrix $\rho(0) = E_1$ 
to the target state $\ketbra{w}{w}$ with a constant fidelity, that is,  
\[
	f(t) := \Tr(\ketbra{w}{w}\rho(t)) = \Omega_n(1),
\]
in time $t = O_n(1/\epsilon_1)$, where $\epsilon_1 = \norm{E_1 w}$ is the shadow of the
target vertex on the principal eigenspace of $G$.

Farhi and Gutmann \cite{fg98} showed that the complete graphs have spatial search (which recovers Grover's
result in the continuous-time setting). They also proved a time lower bound of $\Omega_n(\sqrt{n})$
for any continuous-time quantum algorithm with unit fidelity on vertex-transitive graphs.
As our first result, we strengthen their time lower bound to $\Omega_n(1/\epsilon_1)$ 
which holds for constant fidelity (instead of perfect).
This lower bound justifies the requirement that the optimal time is $O_n(1/\epsilon_1)$.

Chakraborty \etal \cite{cnao16} observed a striking property:
a constant spectral gap $\Delta_2$ is sufficient for optimal spatial search. 
Here, $\Delta_2 = \theta_1 - \theta_2$ is the distance between the two largest normalized eigenvalues of the graph.
In particular, this implies that random graphs exhibit spatial search property almost surely.
But, this does not explain why the $n$-cube has spatial search 
(studied by Childs and Goldstone \cite{cg04}) 
since the spectral gap vanishes as $n$ grows.

Subsequently, Chakraborty \etal \cite{cnr20} improved the observation from \cite{cnao16}
by showing a characterization of optimal spatial search under the assumption of
\begin{equation} \label{eqn:cnr-assumption}
	\epsilon_1 \ll \frac{S_1 S_2}{S_3} \ \ \mbox{ and } \ \ \epsilon_1 \ll \sqrt{S_2}\Delta_2,
\end{equation}
where $S_k = \sum_{r=2}^{d} \norm{E_r w}^2(\theta_1 - \theta_r)^{-k}$, for $k=1,2,3$,
are spectral parameters related to the graph $G$.
Another crucial observation made in \cite{cnr20} is that $S_1$ is the best choice
for the scaling parameter $\gamma$.
As stated in \cite{cnr20}, the unconditional characterization of graphs with optimal spatial search
is a longstanding open question.

In this work, we improve the result of Chakraborty \etal \cite{cnr20} 
by showing a characterization of optimal spatial search under the simpler assumption 
\[
	\epsilon_1 \ll \sqrt{S_1}\Delta_2,
\] 
also under the choice of $\gamma = S_1$.
We show that our assumption is asymptotically similar to the second condition 
in Equation \eqref{eqn:cnr-assumption} for the relevant ranges of interest.
Our improvement is obtained through tighter estimates on the leading eigenvalue perturbations 
derived from a determinant formula of Cauchy.
We also observe a critical condition (hidden in previous analyses) for the strict interlacing between 
pairs of the two largest eigenvalues of the matrices, before and after perturbation.
For this, we explicitly require the second largest eigenvalue of the graph belong to the eigenvalue support 
of the target vertex, and then appeal to Weyl inequalities to provide strict interlacing.

We then apply the characterization to provide new examples of graphs with optimal spatial search 
and to offer alternative proofs for existing families. 
For example, we observe that the Hamming graphs $H(n,q)$, for any constant $q$, have optimal spatial search. 
As a special case, this include the binary $n$-cube $H(n,2)$ which was observed by Childs and Goldstone \cite{cg04}.
For another example, we observe that strongly regular graphs have optimal spatial search since they also have 
constant spectral gap.
This confirms the observation of Janmark \etal \cite{jmw14} obtained using degenerate perturbation theory.

For distance regular graphs with larger diameter, Wong \cite{w16} and then Tanaka \etal \cite{tsp} proved that 
the Johnson graphs $J(n,k)$, for constant $k \ge 3$, have spatial search. 
Since the constant spectral gap condition holds for Johnson graphs, this provides an alternative and immediate
proof that they have optimal spatial search.
Moreover, we also observe that Grassmann graphs and, in fact, most bounded diameter distance-regular graphs 
with classical parameters have optimal spatial search (see Figure \ref{fig:families}).
Both of these are again consequences of the constant spectral gap condition.

Our original motivation for this work was to understand obstructions to optimal spatial search.
To this end, we found a collection of necessary conditions for spatial search which are largely
based on techniques used in \cite{cnr20}. Aside from being a crucial ingredient for proving the 
tight characterization, these necessary conditions provide asymptotic explanations why certain 
families of graphs lack the spatial search property. For example, they can be used to show 
explicitly why cycles lack optimal spatial search -- a well-known folklore result. 
These conditions can potentially be adapted to other classes such as a small Cartesian product of cycles. 

The proofs we employ are elementary as they only use basic tools from matrix theory which
do not appeal to perturbative methods. We nevertheless adopt standard asymptotic arguments commonly 
used in random graphs and complexity of algorithms.


\section{Background}

We assume the standard inner product $\braket{v}{w}$ over $\CC^n$.
All vectors are assumed normalized under the $2$-norm defined by $\norm{w} = \sqrt{\braket{w}{w}}$.
The set of all $n \times n$ matrices with complex entries is denoted $\mat_n(\CC)$.
As with vectors, we define the $2$-norm of a matrix as $\norm{A} = \sqrt{\braket{A}{A}}$,
where $\braket{A}{B} = \Tr(A^\dagger B)$ is the inner product between matrices.
The spectrum $\Sp(A)$ of a matrix $A$ is the set of its eigenvalues.
In this work, we will focus primarily on Hermitian matrices whose eigenvalues are guaranteed to be real.
We adopt the notation $\lambda_i(A)$ to represent the $i$th largest eigenvalue of the matrix $A$.
We call a Hermitian matrix {\em normalized} if its spectrum lies in $[0,1]$ {\em and} $1$ is a simple eigenvalue;
so, in this case $1 = \lambda_1(A) > \lambda_2(A) \ge \ldots \ge \lambda_n(A) \ge 0$.

In stating the spectral decomposition of a Hermitian matrix, say $A = \sum_{r=1}^{d} \theta_r E_r$, 
we always assume the distinct eigenvalues are listed in decreasing order, that is, $\theta_1 > \theta_2 > \ldots > \theta_d$.
Recall that $E_r$ is the orthogonal projection onto the eigenspace corresponding to eigenvalue $\theta_r$
where $E_r$ is Hermitian with $E_r^2 = E_r$, $\sum_{r=1}^{d} E_r = I$, and $E_r E_s = E_r$ if $r=s$ 
and is $0$ otherwise.
The {\em eigenvalue support} of a vector $w \in \CC^n$ with respect to $A$ is defined as
\begin{equation}
\Supp_{A}(w) = \{\theta_r : \norm{E_r w} \neq 0\},
\end{equation} 
which is the set of eigenvalues whose eigenspaces are not fully contained in the subspace $w^\perp$.
For each positive integer $k$, we let
\begin{equation} \label{eqn:moment}
S_k = \sum_{r=2}^{d} \frac{\norm{E_r w}^2}{(\theta_1-\theta_r)^k}.
\end{equation}
These are spectral parameters which will play an important role in characterizing optimal spatial search.
They originally appeared in Childs and Goldstone \cite{cg04}.
Further background on matrix theory may be found in Horn and Johnson \cite{hj13}.

\paragraph{Asymptotics}
We use the standard asymptotic notation to compare the relative order of magnitude of 
two sequences of numbers $f_n$ and $g_n$ depending on a parameter $n \rightarrow \infty$. 
Our main source is Janson \etal \cite{jlr}. We assume $f_n,g_n > 0$ for sufficiently large $n$.
We write:
\begin{itemize}
\item $f_n \lesssim g_n$ or $f_n = O_n(g_n)$ as $n \rightarrow \infty$ 
	if there exist constants $c,n_0 > 0$ such that $f_n \le cg_n$ for $n \ge n_0$.
\item $f_n \gtrsim g_n$ or $f_n = \Omega_n(g_n)$ as $n \rightarrow \infty$ 
	if there exist constants $c,n_0 > 0$ such that $f_n \ge cg_n$ for $n \ge n_0$.
\item $f_n \asymp g_n$ or $f_n = \Theta_n(g_n)$ as $n \rightarrow \infty$ if $f_n = O_n(g_n)$ and $f_n = \Omega_n(g_n)$.
\item $f_n \ll g_n$ or $f_n = o_n(g_n)$ if $f_n/g_n \rightarrow 0$ as $n \rightarrow \infty$.
\item $f_n \gg g_n$ or $f_n = \omega_n(g_n)$ if $g_n/f_n \rightarrow 0$ as $n \rightarrow \infty$.
\item $f_n \sim g_n$ if $f_n/g_n \rightarrow 1$ as $n \rightarrow \infty$.
\end{itemize}
We omit the expression ``as $n \rightarrow \infty$'' when it is clear from context. 
Given that most results are asymptotic, we assume that $n$ is sufficiently large without explicitly stating this.

\medskip
Our focus is on a family of graphs $\{G_n\}_{n=1}^{\infty}$ instead of individual graphs, 
and hence most assertions are asymptotic in nature and will depend on $n$ as it tends to $\infty$.
When it is clear, we simply write $G_n$ for the family of graphs, and even simply $G$ if $n$ is understood 
from context.
Given a graph $G=(V,E)$ that is undirected, its adjacency matrix $A(G)$ is a matrix whose $(u,v)$-entry
is $1$ if $(u,v) \in E$, and is $0$ otherwise. In this work, we will allow the adjacency matrix
be a Hermitian matrix whose nonzero entries are complex valued, and will denote it as $H(G)$. 
For example, this may include the case of signed graphs ($\pm 1$ entries) or complex oriented graphs
($\pm i$ entries).

A continuous-time quantum walk on $G$ with a Hermitian adjacency matrix $H(G)$ is governed by the 
Schr\"{o}dinger equation defined by
\begin{equation} \label{eqn:schrodinger}
	\rho'(t) = -i[H(G), \rho(t)]
\end{equation}
where $\rho(t)$ is a positive semidefinite matrix of unit trace (also called a density matrix).
Here, $[A,B] = AB-BA$ denotes the commutator of two matrices $A,B \in \mat_n(\CC)$.
The solution of the above equation is given by
\[
	\rho(t) = e^{-it H(G)} \rho(0) e^{it H(G)}.
\]
We adopt the density matrix formulation of the Schr\"{o}dinger evolution since it leads to simpler
analyses overall (as global phase factors disappear, for example) and it can be easily generalized
to more realistic settings.
For more background on quantum information, please see Nielsen and Chuang \cite{nc}

Given that spatial search is heavily influenced by the seminal work \cite{g97},
in the rest of this paper we will call a graph with the optimal spatial search property {\em Groverian}.

\begin{definition} \label{def:spatial-search} 
Let $\{G_n\}$ be a family of graphs with normalized adjacency matrix $H(G_n)$ where $E_1$ is 
the orthogonal projection onto its principal eigenspace.
We say $G_n$ is $\gamma$-{\em Groverian} if for any $w \in \CC^n$ with $\norm{E_1 w} \neq 0$, 
the density matrix evolution defined by
\[
	\rho(t) = U(t) \rho(0) U(t)^{-1},
	\ \hspace{.3in} \
	\mbox{ with $\rho(0) = E_1$},
\]
where $U(t) = \exp(-it(\gamma H(G_n) + \ketbra{w}{w}))$,
there is a time $\tau = O_n(1/\norm{E_1 w})$ so that the fidelity satisfies
\begin{equation} \label{eqn:def-spatial-search}
	f(\tau) := \Tr(\ketbra{w}{w} \rho(\tau)) = \Omega_n(1).
\end{equation}
The family of graphs is Groverian if it is $\gamma$-{Groverian} for some $\gamma > 0$.
\end{definition}

The condition on time is the {\em quadratic speedup} requirement that is the hallmark signature of Grover search.
This is because the probability of measuring $w$ given the state $E_1$ is $\norm{E_1 w}^2$, 
and hence generating $w$ has a geometric time of $1/\norm{E_1 w}^2$. 

For a spatial search algorithm to be fully constructive, most previous works require the graphs be vertex transitive.
This assumption was made in Farhi and Gutmann \cite{fg98} and in Childs and Goldstone \cite{cg04}.
This guarantees that the choice of the scaling parameter $\gamma$ and the time $t$ in Definition \ref{def:spatial-search} 
are not dependent on the vertex $w$. But, as pointed out by Meyer and Wong \cite{mw15}, we can also allow graphs whose 
automorphism group has a constant number of orbits. They observed that the search algorithm can simply 
check each orbit separately as most of the relevant spectral parameters are constant within each orbit\footnote{A pertinent 
example given in \cite{mw15} is the {\em barbell} graph 
obtained from connecting two disjoint cliques $K_n$ by a single edge; this graph has two orbits with sizes
$2$ and $2n-2$, respectively.}.
Note that a vertex transitive graph has only a single orbit.
In this work, we will place the same assumption on our graphs.

Given that most of our statements hold for Hermitian matrices, we will view graphs largely
through their Hermitian adjacency matrices. To that end, we fix a convenient terminology to
capture a triplet of a Hermitian matrix, a unit norm vector and a positive scalar that will play
a central role in all of our assertions.

\begin{assumption} \label{assume:star}
We call $(H,w,\gamma)$ a {\em tuplet} if the following (notational) assumptions hold.
\begin{enumerate}[\hspace{.2in}(a)]
\item $H \in \mat_n(\CC)$ is a normalized Hermitian matrix whose spectral decomposition is $H = \sum_{r} \theta_r E_r$.
	Recall that as $H$ is normalized, its eigenvalues lie in $[0,1]$ and $1$ is a simple eigenvalue.

\item $w \in \CC^n$ is a vector with unit norm which satisfies $\norm{E_1 w} \neq 0$ and $\norm{E_1 w} = o_n(1)$.
	Whenever it is clear from context, we will use the abbreviated notation $\epsilon_1 := \norm{E_1 w}$.

\item $\gamma \in \RR$ is a positive scalar whereby the perturbed matrix $\gamma H + \ketbra{w}{w}$ has 
	the spectral decomposition $\sum_{p} \zeta_p F_p$.
\end{enumerate}
\end{assumption}

\section{Time Lower Bounds}

We motivate the condition on time in Definition \ref{def:spatial-search}. 
Farhi and Gutmann \cite{fg98} proved that spatial search with fidelity $1$ on a vertex-transitive graph 
requires time $\Omega_n(\sqrt{n})$. 
We extend their result to show a time lower bound of $\Omega_n(1/\eone)$ for constant fidelity 
which applies to arbitrary graphs. 
Note that $\eone = 1/\sqrt{n}$ for vertex-transitive graphs whenever $w$ denote the characteristic
vector of a vertex.
This matching lower bound justifies the choice of the optimal time.

\begin{theorem} \label{thm:density-time-lower-bound}
Let $(H,w,\gamma)$ be a tuplet where $H$ is $\gamma$-Groverian at time $\tau$.
Then, $\tau = \Omega_n(1/\eone)$.
\end{theorem}

\begin{proof}
The first part of the proof follows \cite{fg98} but in the density matrix language.
Let $H_w = H_0 + \ketbra{w}{w}$, where $H_0 = \gamma H$.
We compare two density matrix evolutions given by
\begin{equation} \label{eqn:density-evols}
	\rho'_w(t) = -i[H_w,\rho_w(t)],
	\ \ \
	\rho'_0(t) = -i[H_0,\rho_0(t)]
\end{equation}
with $\rho_w(0) = \rho_0(0) = E_1$.
Note $\rho_0(t) = E_1$ for all $t$.
Assume for now that the fidelity is one or $\rho_w(\tau) = \ketbra{w}{w}$. 
We will remove this assumption later.

The proof proceeds by analyzing bounds on $\norm{\rho_w(t) - \rho_0(t)}^2$.
First, by taking derivative, we have
\[
	\dt \norm{\rho_w(t) - \rho_0(t)}^2  
	= -2\ip{\rho_w(t)}{\rho_0(t)}'.
\]
From the product rule and Equation \eqref{eqn:density-evols}, we see that
\begin{eqnarray*}
\ip{\rho_w(t)}{\rho_0(t)}'
	& = & \ip{\rho_w(t)}{\rho'_0(t)} + \ip{\rho'_w(t)}{\rho_0(t)} \\
	& = & i \ip{[H_w, \rho_w(t)]}{E_1} \\
	& = & i \braket{w}{[E_1, \rho_w(t)]w}.
\end{eqnarray*}
Now, notice that
\begin{eqnarray*}
|\braket{w}{[E_1, \rho_w(t)]w}|
	& \le & 2|\braket{w}{E_1\rho_w(t)w}|,
		\ \mbox{ as $|x^\dagger [A,B]x| \le 2|x^\dagger (AB)x|$} \\
	& \le & 2\norm{E_1 w} \norm{\rho_w(t)w},
		\ \mbox{ since $|\braket{a}{b}| \le \norm{a}\norm{b}$} \\
	& \le & 2\eone,
		\ \hspace{.2in} \mbox{ because $\norm{\rho_w(t)w} \le 1$.}
\end{eqnarray*}
Putting these together, we get
\[
	\left|\dt \norm{\rho_w(t) - \rho_0(t)}^2\right| 
	\ \le \ 4 \eone.
\]
By the Fundamental Theorem of Calculus, we obtain
\begin{equation} \label{eqn:density-alt-ubound}
	\norm{\rho_w(\tau) - \rho_0(\tau)}^2 
	\ = \ \int_0^{\tau} \left(\dt \norm{\rho_w(t) - \rho_0(t)}^2\right) dt 
	\ \le \ 4\eone \tau.
\end{equation}
Since $\rho_w(\tau) = \ketbra{w}{w}$, we have
$\norm{\rho_w(\tau) - \rho_0(\tau)}^2 = 2(1-\eone^2)$, and therefore
\[
	1-\eone^2 
	\ \le \ 2\eone \tau,
\]
which yields the lower bound $\tau = \Omega_n(1/\eone)$.

Finally, we remove the assumption $\rho_w(\tau) = \ketbra{w}{w}$. 
Suppose that $\norm{\rho_w(\tau) - \ketbra{w}{w}}^2 \le \delta$, for some $\delta \in (0,1)$.
Our strategy is to reduce this case to the former case by using the following inequality.

\begin{fact} \label{fact:density-squared-triangle}
(Triangle Inequality for Squared Norm)
For matrices $A,B,C \in \mat_n(\CC)$, 
\[
	\norm{A-C}^2 \le 2\norm{A-B}^2 + 2\norm{B-C}^2.
\]
\end{fact}
\begin{proof}
Squaring the triangle inequality, we get 
\[
	\norm{A-C}^2 \le \norm{A-B}^2 + \norm{B-C}^2 + 2\norm{A-B}\norm{B-C}.
\]
Now, observe $0 \le (\norm{A-C} - \norm{B-C})^2$. 
\end{proof}

Applying Fact \ref{fact:density-squared-triangle}, we see that
\[
	2\norm{\rho_w(\tau) - \rho_0(\tau)}^2
	\ge \norm{\ketbra{w}{w} - \rho_0(\tau)}^2 - 2\norm{\rho_w(\tau) - \ketbra{w}{w}}^2 
	\ge \norm{\ketbra{w}{w} - \rho_0(\tau)}^2 - 2\delta.
\]
Thus, we have
\[
	\norm{\rho_w(\tau) - \rho_0(\tau)}^2 
		\ \ge \ (1 - \eone^2 - \delta)
\]
which can be combined with Equation \eqref{eqn:density-alt-ubound} to obtain $\tau = \Omega_n(1/\eone)$.
\end{proof}


\section{Interlacing}

We review some relevant tools from matrix theory and prove a few preliminary results.

\subsection{Weyl Inequalities}

A theorem of Weyl on eigenvalue interlacing is key to our analysis. 
Given that we restate the theorem slightly, we prove it for completeness.

\begin{lemma} (Subspace intersection, Lemma 4.2.3 in \cite{hj13}) \\
Let $W_1,\ldots,W_k$ be subspaces of $\CC^n$ and let $d = \dim W_1 + \ldots + \dim W_k - (k-1)n$.
If $d \ge 1$, then $ \dim(\bigcap_{i=1}^{k} W_i) \ge d$. 
In particular, there is a unit vector in $W_1 \cap \ldots \cap W_k$.
\end{lemma}

\begin{proof}
Note that $\dim(W_1 \cap W_2) + \dim(W_1 + W_2) = \dim W_1 + \dim W_2$, which implies
$\dim(W_1 \cap W_2) \ge \dim W_1 + \dim W_2 - n$. So, if $\dim W_1 + \dim W_2 - n \ge 1$, then
$W_1 \cap W_2$ contains a nonzero vector. The claim follows by induction.
\end{proof}

\begin{theorem} \label{thm:weyl} (Weyl Interlacing, restatement of Lemma 4.3.1 in \cite{hj13}) \\
Let $A,B$ be two $n \times n$ Hermitian matrices. 
For $i=1,\ldots,n$, we have
\begin{equation} \label{eqn:weyl1}
\lambda_i(A+B) \ \le \ \lambda_{i-j}(A) + \lambda_{j+1}(B),
	\ \ \mbox{ $j=0,\ldots,i-1$}
\end{equation}
with equality if and only if there is $v \neq 0$ so that $(A+B)v = \lambda_i(A+B)v$,
$Av = \lambda_{i-j}(A)v$, and $Bv = \lambda_{j+1}(B)v$,
and we have
\begin{equation} \label{eqn:weyl2}
\lambda_{i+j}(A) + \lambda_{n-j}(B) \ \le \ \lambda_{i}(A+B),
	\ \ \mbox{ $j=0,\ldots,n-i$}
\end{equation}
with equality if and only if there is $v \neq 0$ so that $(A+B)v = \lambda_i(A+B)v$,
$Av = \lambda_{i+j}(A)v$, and $Bv = \lambda_{n-j}(B)v$.
\end{theorem}

\begin{proof}
For $i=1,\ldots,n$, let $z_i$, $x_i$, and $y_i$ be orthonormal eigenvectors of $A$, $B$, and $A+B$, respectively, 
corresponding to their $i$th largest eigenvalues.

We define 
$W_1 = \spn\{z_{i-j},\ldots,z_n\}$,
$W_2 = \spn\{x_{j+1},\ldots,x_n\}$, 
and
$W_3 = \spn\{y_1,\ldots,y_i\}$.
Let $d_1 = \dim W_1 = n-i+j+1$, $d_2 = \dim W_2 = n-j$, and $d_3 = \dim W_3 = i$.
Since $\dim(W_1 \cap W_2 \cap W_3) = d_1 + d_2 + d_3 - \dim(W_1+W_2+W_3) \ge d_1 + d_2 + d_3 - 2n=1$,
the subspace $W_1 \cap W_2 \cap W_3$ contains a nonzero vector $v$.
For any nonzero $v \in W_1 \cap W_2 \cap W_3$, we have
\begin{equation}
\lambda_i(A+B) \le v^\dagger(A+B)v 
    \le \lambda_{i-j}(A) + \lambda_{j+1}(B),
    \ \mbox{ for $j=0,\ldots,i-1$. }
\end{equation}
Equality is achieved for $i,j$ if and only if there is a nonzero $v \in W_1 \cap W_2 \cap W_3$
for which $(A+B)v = \lambda_i(A+B)v$, $Av = \lambda_{i-j}(A)v$, and $Bv = \lambda_{j+1}(B)v$.

For the second inequality, observe that $\lambda_i(-A) = -\lambda_{n-i+1}(A)$. Therefore,
\begin{equation}
-\lambda_{n-i+1}(A+B) = \lambda_i(-A-B) \le \lambda_{i-j}(-A) + \lambda_{j+1}(-B) = -\lambda_{n-i+j+1}(A) - \lambda_{n-j}(B)
\end{equation}
which implies
\begin{equation}
\lambda_{n-i+j+1}(A) + \lambda_{n-j}(B) \le \lambda_{n-i+1}(A+B).
\end{equation}
Now, rename $n-i+1$ to $i$ and keep $j$ (and hence $n-i+j+1$ to $i+j$). This yields
\begin{equation}
\lambda_{i+j}(A) + \lambda_{n-j}(B) \le \lambda_{i}(A+B),
    \ \mbox{ for $j=0,\ldots,n-i$. }
\end{equation}
Equality is achieved for $i,j$ if and only if there is a nonzero $v \in W_1 \cap W_2 \cap W_3$
for which $(A+B)v = \lambda_i(A+B)v$, $Av = \lambda_{i+j}(A)v$, and $Bv = \lambda_{n-j}(B)v$ with
$W_1 = \spn\{z_1,\ldots,z_{i+j}\}$,
$W_2 = \spn\{x_1,\ldots,x_{n-j}\}$,
and
$W_3 = \spn\{y_i,\ldots,y_n\}$.
\end{proof}

\begin{lemma} \label{lemma:perturb-me} (Strict Interlacing)
Let $H$ be a normalized $n \times n$ Hermitian matrix with spectral decomposition $H=\sum_r\theta_r E_r$.
Let $w \in \CC^n$ be a vector with unit norm where $\theta_1,\theta_2 \in\Supp_H(w)$.
Then for any $\gamma>0$, the two largest eigenvalues of $\gamma H + \ketbra{w}{w}$ are 
simple and they strictly interlace the two largest eigenvalues of $\gamma H$.
\end{lemma}

\begin{proof}
We apply Theorem \ref{thm:weyl} with $A=\gamma H$ and $B = \ketbra{w}{w}$.
Note $\lambda_1(H)=\lambda_1(\ketbra{w}{w})=1$ are simple eigenvalues, and $\lambda_j(\ketbra{w}{w}) = 0$ for $j=2,\ldots,n$.
From Equation \eqref{eqn:weyl2}, with $i=1$ and $j=0$, we get
\begin{equation} \label{eqn:delta-plus}
\lambda_1(\gamma H) + \lambda_n(\ketbra{w}{w}) = \lambda_1(\gamma H) \le \lambda_1(\gamma H + \ketbra{w}{w}).
\end{equation}
Similarly, from Equation \eqref{eqn:weyl1}, with $i=2$ and $j=1$, we get
\begin{equation} \label{eqn:delta-minus}
\lambda_2(\gamma H + \ketbra{w}{w}) \le \lambda_1(\gamma H) + \lambda_2(\ketbra{w}{w}) = \lambda_1(\gamma H).
\end{equation}
Both of the inequalities above are strict since $\lambda_1(H)$ is simple and $E_1 w \neq 0$ implies $\ketbra{w}{w}z_1\neq0$ where $z_1$ is the principal eigenvector with $\ketbra{z_1}{z_1}=E_1$.

Next, we apply Theorem \ref{thm:weyl} with $i=2$ and $j=0$ in Equation \eqref{eqn:weyl2}, to get
\begin{equation} \label{eqn:theta-minus-simple}
\lambda_2(\gamma H) + \lambda_{n}(\ketbra{w}{w}) = \lambda_2(\gamma H) \le \lambda_2(\gamma H + \ketbra{w}{w}).
\end{equation}
Since $\theta_2\in\Supp_H(w)$, there exists an eigenvector $z_2$ corresponding to $\lambda_2(H)$ so that $\ketbra{w}{w}z_2\neq0$. For any nonzero vector $v$ from the subspace $W_1\cap W_2\cap W_3$ (as constructed in the proof of Theorem \ref{thm:weyl}), 
with $W_1 = \spn\{\ket{z_1},\ket{z_2}\}$, we always have $\ketbra{w}{w}z_2\neq0$
and hence the inequality in Equation \eqref{eqn:theta-minus-simple} is strict.
\end{proof}


\subsection{Cauchy's Equality}

The next formula due to Cauchy can be derived from the determinant of bordered matrices (see \cite{hj13}).

\begin{lemma} (Cauchy) \label{lemma:cauchy}
For any $n \times n$ matrix $A$, any vectors $x,y\in\CC^n$, we have
\begin{equation} \label{eqn:cauchy0}
\det(A + xy^{\dagger}) = \det(A) + y^{\dagger} \Adj(A) x.
\end{equation}
Moreover, if $A$ is nonsingular, then 
\begin{equation} \label{eqn:cauchy}
\det(A + xy^{\dagger}) = \det(A) (1 + y^{\dagger} A^{-1}x).
\end{equation}
\end{lemma}

We apply the above lemma to provide sharper estimates for the case when $\gamma=S_1$ 
on the locations of the two largest perturbed eigenvalues relative to the unperturbed principal eigenvalue.

\begin{proposition} \label{prop:deltas}
Let $(H,w,S_1)$ be a tuplet. Then
\begin{equation} \label{eqn:cauchy-plus}
	\epsilon_1^2 \ < \ \zeta_1 - S_1 \ \lesssim \ \sqrt{S_1} \epsilon_1.
\end{equation}
Moreover, if $\theta_2 \in \Supp_H(w)$, then
\begin{equation} \label{eqn:cauchy-minus}
	0 \ < \ S_1 - \zeta_2 \ \lesssim \ \sqrt{S_1} \epsilon_1. 
\end{equation}
\end{proposition}

\begin{proof}
Let $\tilde{H} = S_1 H + \ketbra{w}{w}$ and $\Delta_r = 1 - \theta_r$, for all $r$.
Note that $\phi(\tilde{H},t) = \det((tI - S_1 H) - \ketbra{w}{w})$. 
Applying Lemma \ref{lemma:cauchy} and assuming $m_r$ is the multiplicity of $\theta_r$, we get
\begin{eqnarray*}
\phi(\tilde{H},t) 
	& = & \det(tI - S_1 H)(1 - {w^{\dagger}}{(tI - S_1 H)^{-1}w}) \\
	& = & \left(\prod_r (t - S_1\theta_r)^{m_r}\right) \left(1 - \sum_r \frac{\norm{E_r w}^2}{(t - S_1\theta_r)}\right).
\end{eqnarray*}
The bound $\zeta_1 - S_1 > \epsilon_1^2$ follows by verifying that $\phi(\tilde{H}, S_1 + \epsilon_1^2)$ is negative.
Observe that if $\phi(\tilde{H}, S_1 + \beta) > 0$, 
then $\zeta_1 - S_1 \le \beta$.
We write the condition $\phi(\tilde{H}, S_1 + \beta) > 0$ as
\[
	1 
	\ > \ \sum_r \frac{\norm{E_r w}^2}{S_1\Delta_r + \beta}
	\ = \ \frac{\epsilon_1^2}{\beta} + \sum_{r \neq 1} \frac{\norm{E_r w}^2}{S_1\Delta_r + \beta}.
\]
Because $\Delta_r \le 1$, for each $r$, it follows that $(1 + \beta/S_1)^{-1}$ is an upper bound 
for the sum $\sum_{r \neq 1} \norm{E_r w}^2/(S_1\Delta_r + \beta)$.
So, to satisfy the previous inequality, it suffices to require
\[
	1 \ > \ \frac{\epsilon_1^2}{\beta} + \frac{1}{1 + \beta/S_1}.
\]
After straightforward calculations, we obtain 
$\beta^2 - \epsilon_1^2 \beta - S_1\epsilon_1^2 > 0$. The roots of this quadratic equation are given by
\[
	\beta_\pm = \frac{1}{2}(\epsilon_1^2 \pm \sqrt{\epsilon_1^4 + 4S_1\epsilon_1^2}).
\]
The positive root is given by $\beta_{+} = O_n(\sqrt{S_1}\epsilon_1)$, since $\epsilon_1 \ll 1$,
which proves Equation \eqref{eqn:cauchy-plus}.

We denote $\delta_{-} = \zeta_2 - S_1$.
By Lemma \ref{lemma:perturb-me}, $|\delta_{-}| > 0$.
Now, observe that if $\phi(\tilde{H}, S_1 - \beta) > 0$ then $|\delta_{-}| \le \beta$.
As above, we write the condition $\phi(\tilde{H}, S_1 - \beta) > 0$ as
\[
	1 
	\ < \ \sum_r \frac{\norm{E_r w}^2}{S_1\Delta_r - \beta}
	\ = \ -\frac{\epsilon_1^2}{\beta} + \sum_{r \neq 1} \frac{\norm{E_r w}^2}{S_1\Delta_r - \beta}.
\]
This yields
\[
	\sum_{r \neq 1} \frac{\norm{E_r w}^2}{S_1\Delta_r} \frac{1}{1 - \beta/(S_1\Delta_r)}
	\ > \ 
	1 + \frac{\epsilon_1^2}{\beta}.
\]
As $\Delta_r \le 1$, for $r \ge 2$, we note that $1/(1 - \beta/S_1)$ is a lower bound
for the expression on the left-hand side. Therefore, we may require instead
\[
	\frac{1}{1 - \beta/S_1} \ > \ 1 + \frac{\epsilon_1^2}{\beta}
\]
which simplifies to
$\beta^2 + \epsilon_1^2\beta - S_1\epsilon_1^2 \ > \ 0$.
The maximum root of the quadratic polynomial is given by
\[
	\beta = \frac{1}{2}(-\epsilon_1^2 + \sqrt{\epsilon_1^4 + 4S_1\epsilon_1^2})
		\ \sim \ \sqrt{S_1}\epsilon_1
\]
which proves Equation \eqref{eqn:cauchy-minus}.
\end{proof}


\section{Necessary Gaps} \label{section:necessary}

We describe necessary conditions for graphs to be Groverian, but first we derive some useful
preliminary observations. Our analysis borrows heavily ideas from Chakraborty \etal \cite{cnr20}.

We start by restating the machinery in \cite{cnr20} (specifically, Theorem 4) using our notation 
for the sake of consistency and to point out certain explicit assumptions that are required.
Given a tuplet $(H,w,\gamma)$, for each $p$, we have
$(\gamma H + \ketbra{w}{w})F_p = \zeta_p F_p$ or
\begin{equation} \label{eqn:pre-genesis}
\ketbra{w}{w}F_p = (\zeta_p I - \gamma H)F_p,
\end{equation}
which implies
\begin{equation} \label{eqn:genesis0}
\norm{F_p w}^2 w = (\zeta_p I - \gamma H)F_p w.
\end{equation}
If $\zeta_p \not\in \Sp(\gamma H)$ holds, it clearly guarantees
\begin{equation} \label{eqn:genesis0b}
\frac{F_p w}{\norm{F_p w}^2} = (\zeta_p I - \gamma H)^{-1} w,
\end{equation}
which further yields
\begin{equation} \label{eqn:another-unity}
\frac{1}{\norm{F_p w}^2} = \sum_r \frac{\norm{E_r w}^2}{(\zeta_p - \gamma\theta_r)^2}.
\end{equation}
Moreover, Equation \eqref{eqn:genesis0b} yields an expression of unity
\begin{equation} \label{eqn:original-unity}
1 = \sum_{r=1}^{d} \frac{\norm{E_r w}^2}{\zeta_p - \gamma\theta_r}.
\end{equation}
Returning to Equation \eqref{eqn:genesis0}, after multiplying both sides by $E_1$, we derive
\begin{equation} \label{eqn:genesis1}
\norm{F_p w}^2 E_1 w = (\zeta_p - \gamma) E_1 F_p w.
\end{equation}
By Theorem \ref{thm:weyl}, since $\norm{E_1 w} \neq 0$, we have strict interlacing where $\zeta_2 < \gamma < \zeta_1$.
Furthermore, $\gamma \neq \zeta_r$, for all $r=1,\ldots,n$. 
This allows us to write 
\begin{equation} \label{eqn:substitute}
	E_1 F_p w \ = \ \frac{\norm{F_p w}^2}{\zeta_p - \gamma} E_1 w,
	\ \ \ \mbox{ $p=1,\ldots,m$.}
\end{equation}

Now, analyzing the fidelity of the quantum walk with the Hamiltonian $\tilde{H} = \gamma H + \ketbra{w}{w}$, we get
\begin{equation} \label{eqn:squared-fidelity}
	f(t)
	\ = \ \Tr(w w^\dagger e^{-it\tilde{H}} E_1 e^{it\tilde{H}})
	\ = \ \sum_{p,q} e^{-it(\zeta_p - \zeta_q)} w^\dagger F_p E_1 F_q w.
\end{equation}
After expanding $w^\dagger F_p E_1 F_q w$ to $(w^\dagger F_p E_1)(E_1 F_q w)$ since $E_1^2 = E_1$,
we use Equation \eqref{eqn:substitute} to substitute $E_1 F_q w$ and $(E_1 F_p w)^\dagger$ to get 
\[
	f(t)
	\ = \
	\sum_p e^{-it\zeta_p} \frac{\norm{F_p w}^2}{\zeta_p - \gamma} \sum_q e^{it\zeta_q} \frac{\norm{F_q w}^2}{\zeta_q - \gamma}
		\norm{E_1 w}^2
\]
since $(w^\dagger E_1)(E_1 w) = \norm{E_1 w}^2$. Thus,
\begin{equation} \label{eqn:qwalk-delta-plus}
	f(t)
	\ = \ \eone^2 \left|\sum_{p} e^{-it\zeta_p} \frac{\norm{F_p w}^2}{(\zeta_p - \gamma)}\right|^2.
\end{equation}
By triangle inequality, the above becomes
\begin{equation} \label{eqn:triangle}
	f(t) \ \le \ \eone^2 \left(\sum_{p} \frac{\norm{F_p w}^2}{|\zeta_p - \gamma|}\right)^2.
\end{equation}
Finally, from Equation \eqref{eqn:qwalk-delta-plus} at the time of origin $t=0$, we get another expression of unity,
\begin{equation} \label{eqn:signed-unity}
1 \ = \ \left|\frac{\norm{F_1 w}^2}{\zeta_1 - \gamma} - \sum_{p \ge 2} \frac{\norm{F_p w}^2}{\gamma - \zeta_p}\right|
\end{equation}
since the first term $\norm{F_1 w}^2/(\zeta_1 - \gamma)$ and the subsequent terms $\norm{F_p w}^2/(\zeta_p - \gamma)$, 
for $p \ge 2$, differ in sign because $\zeta_1 > \gamma > \zeta_p$.
As $|A - B| \ge |A| - |B|$ by triangle inequality,
the last equation immediately implies the following pair of inequalities:
\begin{equation} \label{eqn:nudge}
	\sum_{p \ge 2} \frac{\norm{F_p w}^2}{\gamma - \zeta_p} \le \frac{\norm{F_1 w}^2}{\zeta_1 - \gamma} + 1,
	\ \ \
	\frac{\norm{F_1 w}^2}{\zeta_1 - \gamma} \le \sum_{p \ge 2} \frac{\norm{F_p w}^2}{\gamma - \zeta_p} + 1.
\end{equation}

Using the above, we are ready to observe some necessary conditions for a graph $G$ to be Groverian.
In the next theorem, we show some conditions on $\eone$ in relation to the eigenvalue gaps
between the largest unperturbed eigenvalue of $\gamma H(G)$ and the perturbed eigenvalues $\zeta_m$
of $\gamma H(G) + \ketbra{w}{w}$. 
In particular, for a graph $G$ to be Groverian, 
$\eone$ must be asymptotically equal to the first gap $\delta_+ := \zeta_1 - \gamma$,
it must be asymptotically greater or equal to the second gap $\delta_{-} := \gamma - \zeta_2$,
and
it must be asymptotically equal to some gap $\gamma - \zeta_p$, for some $p \ge 2$, 
but not necessarily the second gap.
Recall that by strict interlacing, we know that $\gamma\theta_2 < \zeta_2 < \gamma < \zeta_1$,
where $\theta_2$ is the second largest eigenvalue of $H(G)$.

The first two observations in the following theorem are implicit in the proof of 
Theorem 4 in \cite{cnr20} but it will be useful to restate them in our notation below.
The third observation appears to be new.

\begin{theorem} \label{thm:necessary-generic}
For a tuplet $(H,w,\gamma)$, suppose one of the following conditions holds:
\begin{enumerate}[\hspace{0.2in}(i)]
\item $\eone \not\asymp \zeta_1 - \gamma$ (equivalently, $\eone \ll \zeta_1-\gamma$ or $\eone \gg \zeta_1-\gamma$), or
\item $\eone \ll \gamma - \zeta_2$, or
\item $\gamma - \zeta_{p-1} \ll \eone \ll \gamma - \zeta_p$, for some $p \ge 3$.
\end{enumerate}
Then, $H$ is not $\gamma$-Groverian.
\end{theorem}

\begin{proof}
We treat each condition as a separate case.

\medskip
\par\noindent{\em Case (i)}: Denote $\delta_{+} := \zeta_1 - \gamma$ and assume that $\delta_{+} \not\asymp \epsilon_1$.
From Equation \eqref{eqn:triangle}, we apply Equation \eqref{eqn:nudge} to obtain the following upper bound
\begin{equation} \label{eqn:replace-me}
	f(t)
	\ \le \ \eone^2 \left( \frac{2\norm{F_1 w}^2}{\delta_{+}} + 1 \right)^2
	\ = \ \left( 2\frac{\epsilon_1}{\delta_{+}} \norm{F_1 w}^2 + \epsilon_1 \right)^2,
\end{equation}
which shows that fidelity goes to zero if $\epsilon_1 \ll \delta_{+}$.
For the opposite direction, apply Equation \eqref{eqn:another-unity} with $p=1$ to get
\[
		\frac{1}{\norm{F_1 w}^2} \ = \ \sum_{r=1}^{d} \frac{\norm{E_r w}^2}{(\zeta_1 - \gamma \theta_r)^2}.
\]
But, the sum is bounded from below by its first term, and so
$\norm{F_1 w}^2 \le \delta_{+}^2/\eone^2$.
Returning to Equation \eqref{eqn:replace-me} and using the preceding upper bound, 
\begin{equation} \label{eqn:two-delta-plus-over-epsilon}
	f(t)
	\le \left(\frac{2\delta_{+}}{\eone} + \eone\right)^2,
\end{equation}
which shows that fidelity goes to zero if $\delta_{+} \ll \eone$.

\medskip
\par\noindent{\em Case (ii)}: 
Let $\delta_{-} := \zeta_2 - \gamma$. 
We start with Equation \eqref{eqn:triangle} 
and use Equation \eqref{eqn:signed-unity} to rewrite a term in the upper bound as
\begin{equation} \label{eqn:pre-flip-me}
	\sum_{p} \frac{\norm{F_p w}^2}{|\gamma - \zeta_p|} 
	\ \le \ 1 + 2 \sum_{p \ge 2} \frac{\norm{F_p w}^2}{\gamma - \zeta_p}
	\ \le \ 1 + 2 \frac{(1 - \norm{F_1 w}^2)}{|\delta_{-}|},
\end{equation}
since $|\delta_{-}| \le \gamma - \zeta_p$ for each $p \ge 2$.
Thus, we derive an upper bound on the fidelity,
\begin{equation} \label{eqn:flip-me}
	f(t)
	\le \left(\frac{2\eone}{|\delta_{-}|} (1 - \norm{F_1 w}^2) + \eone\right)^2,
\end{equation}
which shows that fidelity tends to zero if $\eone \ll |\delta_{-}|$.

\medskip
\par\noindent{\em Case (iii)}: 
Starting with Equation \eqref{eqn:pre-genesis}, after multiplying by $E_1$, we derive
\[
	E_1 \ketbra{w}{w} F_p = (\zeta_p - \gamma)E_1 F_p.
\]
Taking a product with itself but for index $q$, we get
\[
	\eone^2 F_q\ketbra{w}{w}F_p = (\zeta_q - \gamma)(\zeta_p - \gamma)F_q E_1 F_p.
\]
Upon taking the trace, this yields
\begin{equation} \label{eqn:flip-trick}
	\eone^2 \norm{F_p w}^2 = (\zeta_p - \gamma)^2 \braket{F_p}{E_1}.
\end{equation}
By Equation \eqref{eqn:squared-fidelity}, we have
\[
	\sum_{p,q} e^{-it(\zeta_p - \zeta_q)} w^\dagger F_p E_1 F_q w
	=
	\sum_{p,q} e^{-it(\zeta_p - \zeta_q)} 
		\frac{\eone}{(\zeta_p - \gamma)} \frac{\eone}{(\zeta_q - \gamma)} \norm{F_p w}^2 \norm{F_q w}^2,
\]
where we have used Equation \eqref{eqn:substitute} twice (for both $p$ and $q$).
Now, note we may apply Equation \eqref{eqn:flip-trick} to ``flip'' one of the ratios as follows:
\[
	\sum_{p,q} e^{-it(\zeta_p - \zeta_q)} 
		\frac{\eone}{\zeta_p - \gamma} \frac{\zeta_q - \gamma}{\eone} \norm{F_p w}^2 \braket{F_q}{E_1}.
\]
This allows us to partition the sums around $m$ based on whether $\epsilon_1 \ll \zeta_p - \gamma$ or $\epsilon_1 \gg \zeta_p - \gamma$.
In all cases, the corresponding terms tend to $0$ as $n \rightarrow \infty$.
\end{proof}

\subsection{Failure around $S_1$}

We describe necessary conditions on $\gamma$ relative to $S_1$.
First, we show failure whenever $\gamma$ is too large relative to $S_1$
(which is an alternate restatement of the first half of Theorem 4 in \cite{cnr20}).

\begin{theorem} \label{thm:necessary-above-s1}
For a tuplet $(H,w,\gamma)$,
if $\eone S_1 \ll \gamma - S_1$, then $H$ is not $\gamma$-Groverian.
\end{theorem}

\begin{proof}
Starting with Equation \eqref{eqn:original-unity} with $p=1$,  we see
\[
	1 \ = \ \sum_{r=1}^{d} \frac{\norm{E_r w}^2}{\zeta_1 - \gamma \theta_r}
	\ = \ \frac{\eone^2}{\delta_{+}} + \sum_{r=2}^{d} \frac{\norm{E_r w}^2}{\gamma(1-\theta_{r}) + \delta_{+}}
	\ \le \ \frac{\eone^2}{\delta_{+}} + \frac{S_1}{\gamma}.
\]
After minor rearrangements, we obtain 
\begin{equation} \label{eqn:delta-plus-upper}
	\delta_{+} 
	\ \le \ \eone^2 \frac{\gamma}{\gamma - S_1}
	\ = \ \eone^2 \left(1 + \frac{S_1}{\gamma - S_1}\right)
	\ \ll \ \eone(1 + o_n(1))
\end{equation}
which proves the claim by appealing to Theorem \ref{thm:necessary-generic}, case {\em (i)}.
\end{proof}

Next, we show failure conditions when $\gamma$ is centered around $S_1$.
We will use this later to show that cycles are not Groverian.

\begin{theorem} \label{thm:necessary-around-s1}
Let $(H,w,\gamma)$ be a tuplet where there is a constant $c > 0$ so that $\norm{E_r w} \le c\eone$, for all $r \ge 2$.
Let 
\[
	I_\alpha = \{r \ge 2: \gamma < \eone^\alpha(1-\theta_r)^{-1}\}.
\]
If $\gamma = S_1(1 + o_n(1))$ and $|I_\alpha| \ge 2c^2$, for some $\alpha \in (1,2)$, 
then $H$ is not $\gamma$-Groverian.
\end{theorem}

\begin{proof}
By Theorem \ref{thm:necessary-generic}, it suffices to show $\delta_{+} \ll \epsilon_1$,
where $\delta_{+} = \zeta_1 - \gamma$. 
Let $\Delta_r = 1-\theta_r$.
Our plan is to prove $\delta_{+} = O_n(\epsilon_1^\alpha)$ for the given $\alpha \in (1,2)$.
By Lemma \ref{lemma:cauchy}, it is enough to show $\phi(H, \gamma + \epsilon_1^\alpha) > 0$ or
equivalently
\[
	1 \ > \ \sum_{r=1}^{d} \frac{\norm{E_r w}^2}{\gamma\Delta_r + \epsilon_1^\alpha}
		\ = \ \epsilon_1^{2-\alpha} + \frac{1}{\gamma} \sum_{r \ge 2} \frac{\norm{E_r w}^2}{\Delta_r + (\epsilon_1^\alpha/\gamma)}.
\]
Let us focus on the last summation. After splitting the summation into two parts, we may bound
each part from above as follows,
\[
	\sum_{r \ge 2} \frac{\norm{E_r w}^2}{\Delta_r + (\epsilon_1^\alpha/\gamma)}
	\le
	\sum_{\substack{r \ge 2:\\\Delta_r < \epsilon_1^\alpha/\gamma}} \frac{\norm{E_r w}^2}{2\Delta_r}
	+
	\sum_{\substack{r \ge 2:\\\Delta_r \ge \epsilon_1^\alpha/\gamma}} \frac{\norm{E_r w}^2}{\Delta_r}.
\]
So, to prove the claim, it suffices to show that
\begin{equation} \label{eqn:penultimate}
	1 
	\ > \
	\epsilon_1^{2-\alpha} + \frac{1}{\gamma} \sum_{r \ge 2} \frac{\norm{E_r w}^2}{\Delta_r}
		- \frac{1}{2\gamma} \sum_{\substack{r \ge 2:\\\Delta_r < \epsilon_1^\alpha/\gamma}} \frac{\norm{E_r w}^2}{\Delta_r}.
\end{equation}
Assume $\gamma = S_1 + \beta$ for some $\beta = o_n(S_1)$. Then, we require that
\begin{equation}
	1 
	\ > \
	\epsilon_1^{2-\alpha} + \frac{1}{1 + (\beta/S_1)} \underbrace{\left(\frac{1}{S_1}\sum_{r \ge 2} \frac{\norm{E_r w}^2}{\Delta_r}\right)}_{=1}
		- \frac{1}{2\gamma} \sum_{\substack{r \ge 2:\\\gamma < \epsilon_1^\alpha/\Delta_r}} \frac{\norm{E_r w}^2}{\Delta_r},
\end{equation}
which is equivalent to
\begin{equation}
	\frac{1}{2\gamma} \sum_{\substack{r \ge 2:\\\gamma < \epsilon_1^\alpha/\Delta_r}} \frac{\norm{E_r w}^2}{\Delta_r} + o_n(1)
	\ > \ 
	\epsilon_1^{2-\alpha}.
\end{equation}
Therefore, we may omit the $o_n(1)$ term as it suffices to satisfy
\[
	\frac{1}{2\gamma} \sum_{\substack{r \ge 2:\\\gamma < \epsilon_1^\alpha/\Delta_r}} \frac{\norm{E_r w}^2}{\Delta_r}
	\ > \ 
	\epsilon_1^{2-\alpha}
	\ \ \mbox{ or } \ \
	\sum_{\substack{r \ge 2:\\\gamma < \epsilon_1^\alpha/\Delta_r}} 
		\frac{\norm{E_r w}^2}{\epsilon_1^2} \frac{\epsilon_1^\alpha}{\gamma\Delta_r}
	\ > \ 
	2.
\]
If $\epsilon_1 \le c\norm{E_r w}$, for all $r \ge 2$, 
the claim follows as $|I_\alpha| \ge 2c^2$.
\end{proof}

\subsection{Near Perfect Fidelity}

When the fidelity is $1-o_n(1)$, we show that the principal perturbed eigenspace must have a significant
overlap with the subspace spanned by the principal eigenspace and the target state.

\begin{theorem}
Let $(H,w,\gamma)$ be a tuplet.
If $H$ is $\gamma$-Groverian with fidelity $1-o_n(1)$, then 
the principal eigenvector of $\gamma H + \ketbra{w}{w}$ is an equal superposition of
$w$ and the principal eigenvector of $H$, up to $o_n(1)$ terms.
\end{theorem}

\begin{proof}
Let $z_1$ and $y_1$ be the principal eigenvectors of $H$ and $\gamma H + \ketbra{w}{w}$, respectively.
Thus, $E_1 = z_1 z_1^\dagger$ and $F_1 = y_1 y_1^\dagger$.
Starting with Equation \eqref{eqn:replace-me} and then applying Equation \eqref{eqn:substitute} with $p=1$, we get
\[
	f(t) 
	\ \le \ 4 \frac{\norm{F_1 w}^4}{\delta_{+}^2}\norm{E_1 w}^2 
	\ = \ 4\norm{E_1 F_1 w}^2
	\ = \ 4|\braket{y_1}{z_1}|^2 |\braket{y_1}{w}|^2.
\]
If $f(t) = 1 - o_n(1)$, we obtain $1-o_n(1) \le 2|\braket{y_1}{z_1}||\braket{y_1}{w}|$.
Since $|\braket{y_1}{z_1}|^2 + |\braket{y_1}{w}|^2 \le 1+o_n(1)$, we have
$0 \le (|\braket{y_1}{z_1}| - |\braket{y_1}{w}|)^2 =  o_n(1)$, which proves the claim.
\end{proof}


\section{Sufficient Gap}

In this section, we prove our optimal characterization of $S_1$-Groverian graphs under a simpler assumption.
First, we justify that the choice of $\gamma = S_1$ is almost best possible.
By Equation \eqref{eqn:original-unity} with $p=1$, we have
\[
	1 = \frac{\eone^2}{\delta_+} + \sum_{r \neq 1} \frac{\norm{E_r w}^2}{\delta_+ + \gamma\Delta_r} 
	\asymp \eone + \sum_{r \neq 1} \frac{\norm{E_r w}^2}{\delta_+ + \gamma\Delta_r}
\]
as $\eone \asymp \delta_+$ 
is a necessary condition due to Theorem \ref{thm:necessary-generic}$(i)$.
Recall that $f_n \asymp g_n$ denotes that $f_n$ and $g_n$ are within constant factors of each other.
Therefore,
\[
	1-\eone \asymp 
	\frac{1}{\gamma} \sum_{r \neq 1} \frac{\norm{E_r w}^2}{\Delta_r} \frac{1}{1 + (\eone/\gamma\Delta_r)}.
\]
If $\gamma \gtrsim \sqrt{S_1}$, since $\eone \ll \sqrt{S_1}\Delta_2$, we have $\eone/\gamma\Delta_r = o_n(1)$. 
This shows
\[
	\gamma(1-\eone) \asymp \sum_{r \neq 1} \frac{\norm{E_r w}^2}{\Delta_r} \frac{1}{1 + (\eone/\gamma\Delta_r)}
	= S_1(1 + o_n(1)).
\]
As $\eone = o_n(1)$, we have $\gamma \asymp S_1$.

The next theorem is our main characterization for optimal spatial search for $\gamma = S_1$.

\begin{theorem} \label{thm:spatial-search}
Let $(H,w,\gamma)$ be a tuplet where $\theta_2 \in \Supp_H(w)$.
If $\eone \ll \sqrt{S_1}\Delta_2$, then
$H$ is $S_1$-Groverian if and only if $S_2/S_1^2 \asymp 1$.
\end{theorem}

\begin{proof}
By Lemma \ref{lemma:perturb-me}, we know that $\zeta_1$ and $\zeta_2$ are simple and $\zeta_1,\zeta_2 \not\in \Sp(S_1 H)$. 
Therefore, let $y_1$ and $y_2$ be their corresponding eigenvectors. 
Define $\delta_{+} = \zeta_1 - S_1$ and $\delta_{-} = \zeta_2 - S_1$. Also, let $\Delta_r = 1-\theta_r$.
Using Equation \eqref{eqn:original-unity} with $p=1,2$, we derive
\begin{equation} \label{eqn:unity0c}
1 \ = \ \sum_{r=1}^{d} \frac{\norm{E_r w}^2}{\zeta_p - S_1\theta_r} 
	\ = \ \frac{\epsilon_1^2}{\delta_\pm} + \sum_{r=2}^{d} \frac{\norm{E_r w}^2}{S_1\Delta_r} \frac{1}{1 + \delta_\pm/(S_1\Delta_r)}.
\end{equation}
Note $(1+\alpha)^{-1} = 1 - \alpha + \alpha^2 (1+\alpha)^{-1}$ holds for all $\alpha \neq -1$.
Using $\alpha = \delta_\pm/(S_1\Delta_r)$, 
after reorganizing and cancelling terms, we arrive at
\begin{equation} \label{eqn:base-point0}
\frac{\epsilon_1^2}{\delta_\pm^2} 
	\ = \ \frac{1}{S_1^2}
		\sum_{r=2}^{d} \frac{\norm{E_r w}^2}{\Delta_r^2} \frac{S_1\Delta_r}{\delta_\pm + S_1\Delta_r}.
\end{equation}
Thus far, this is similar to the first portion of the proof of Theorem 2 in \cite{cnr20}. 
But, now we exploit the sharper estimates given by Proposition \ref{prop:deltas} which will 
considerably simplify the rest of our proof.

\begin{lemma} \label{lemma:delta-plus-minus}
$\epsilon_1^2/\delta_{\pm}^2 \sim S_2/S_1^2$.
\end{lemma}

\begin{proof} 
Proposition \ref{prop:deltas} shows that $|\delta_{\pm}| = O_n(\sqrt{S_1}\epsilon_1)$,
and as $\epsilon_1 \ll \sqrt{S_1}\Delta_2$, we get $|\delta_{\pm}| \ll S_1\Delta_2$.
Applying this to Equation \eqref{eqn:base-point0} proves the claim.
\end{proof}

\begin{lemma} \label{lemma:full-expanse}
$2\epsilon_1^2/\delta_{+}^2 \sim 1/\norm{F_1 w}^2$
and
$2\epsilon_1^2/\delta_{-}^2 \sim 1/\norm{F_2 w}^2$.
\end{lemma}

\begin{proof}
We apply Equation \eqref{eqn:another-unity} to $\zeta_1$ and $\zeta_2$, and after some rearrangements, we get
\begin{equation} \label{eqn:local}
\frac{1}{\norm{F_p w}^2} 
	= \frac{\epsilon_1^2}{\delta_\pm^2} 
		+ \frac{1}{S_1^2}\sum_{r=2}^{d} \frac{\norm{E_r w}^2}{\Delta_r^2} \frac{1}{(1+\delta_\pm/(S_1\Delta_r))^2}
	\ \ \ \mbox{ ($p=1,2$).}
\end{equation}
Proposition \ref{prop:deltas} shows that $|\delta_{\pm}| = O_n(\sqrt{S_1}\epsilon_1)$,
which combined with the assumption $\epsilon_1 \ll \sqrt{S_1}\Delta_2$ yields $|\delta_\pm| \ll S_1\Delta_2$.
Applying this to Equation \eqref{eqn:local}, we obtain
\begin{equation}
\frac{1}{\norm{F_p w}^2} 
	\ = \ \frac{\epsilon_1^2}{\delta_\pm^2} + \frac{S_2}{S_1^2} (1+o_n(1))
	\ \sim \ \frac{2\epsilon_1^2}{\delta_\pm^2}
	\ \ \ \mbox{ ($p=1,2$),}
\end{equation}
as $\epsilon_1^2/\delta_\pm^2 \sim S_2/S_1^2$ by Lemma \ref{lemma:delta-plus-minus}.
\end{proof}

The use of density matrices simplifies the proof of the following result
as most arguments involving phase factors are no longer necessary.

\begin{lemma} \label{lemma:propositive}
If $t = O_n(\sqrt{S_2}/S_1\epsilon_1)$, then $f(t) = \Omega_n(S_1/\sqrt{S_2})$.
\end{lemma}

\begin{proof}
By Equation \eqref{eqn:genesis1}, we have
$E_1 F_p w = \norm{F_p w}^2 E_1 w/(\zeta_p - S_1)$.
For $p=1,2$, 
using Lemma \ref{lemma:full-expanse} to replace $\norm{F_p w}^2$, we obtain
\begin{equation} \label{eqn:full-support}
	\norm{E_1 F_p w} 
	= \norm{F_p w}^2 \frac{\eone}{|\delta_{\pm}|} 
	\sim \frac{|\delta_\pm|}{2\epsilon_1}.
\end{equation}
Let $\delta = (\zeta_1 - \zeta_2)/2$. Then, the fidelity is given by
\begin{eqnarray*}
f(t) 
	& = & \Tr(\ketbra{w}{w} \sum_{p,q} e^{-it(\zeta_p - \zeta_q)} F_p E_1 F_q) \\
	& = & \sum_{p,q} e^{-it(\zeta_p - \zeta_q)} \braket{E_1 F_p w}{E_1 F_q w} \\
	& = & \sum_{p,q} e^{-it(\zeta_p - \zeta_q)} \eone^2
		\frac{\norm{F_p w}^2}{\zeta_p - S_1} \frac{\norm{F_q w}^2}{\zeta_q - S_1} \\
	& = & \frac{1}{2}(1-\cos(2\delta t))\frac{\delta^2}{\epsilon_1^2} +
		\epsilon_1^2 \sum_{p,q \neq 1,2} e^{-it(\zeta_p - \zeta_q)} 
			\frac{\norm{F_p w}^2}{S_1 - \zeta_p} \frac{\norm{F_q w}^2}{S_1 - \zeta_q} \\
	& \ge & \frac{\delta^2}{\epsilon_1^2},
		\ \ \ \mbox{ by setting $t = \frac{\pi}{2}\delta^{-1}$.}
\end{eqnarray*}
Now, note that $\delta = \frac{1}{2}(\delta_{+} + |\delta_{-}|) \asymp \epsilon_1$.
\end{proof}

We have shown that if $S_1^2/S_2 = \Theta_n(1)$ then $H$ is $S_1$-Groverian.
It remains to show the converse. 

\begin{lemma} \label{lemma:contrapositive}
If $f(t) = \Omega_n(1)$, for some $t = O_n(1/\epsilon_1)$, then $S_2/S_1^2 = \Theta_n(1)$.
\end{lemma}

\begin{proof}
We prove the contrapositive. If $S_2/S_1^2 \not\asymp 1$, or equivalently, $\epsilon_1^2 \not\asymp \delta_{+}^2$ 
by Lemma \ref{lemma:delta-plus-minus}, then $H$ is not Groverian by Theorem \ref{thm:necessary-generic}$(i)$. 
Note $\epsilon_1^2 \asymp \delta_{+}^2$ if and only if $\epsilon_1 \asymp \delta_{+}$.
\end{proof}

Lemma \ref{lemma:propositive} and \ref{lemma:contrapositive} completes the proof of the theorem.
\end{proof}

\medskip
\par\noindent{\bf Remarks}.
We now compare Theorem \ref{thm:spatial-search} in the context of the known results in \cite{cnao16,cnr20}.
Notice that Theorem \ref{thm:spatial-search} implies the main result (Lemma 1) in Chakraborty \etal \cite{cnao16}.
Their result requires the {constant gap condition} $\Delta_2 = \Omega_n(1)$ which is a stronger assumption 
because of the following observation. 
In what follows, we denote $\epsilon_2 = \norm{E_2 w}$.

\begin{fact} \label{fact:constant-gap}
Suppose $\epsilon_1 = o_n(1)$. If $\Delta_2 = \Omega_n(1)$, then $S_1\Delta_2 = \Theta_n(1)$.
\end{fact}

\begin{proof}
Note $\epsilon_2^2 + \Delta_2(1-\epsilon_1^2-\epsilon_2^2) \le S_1\Delta_2 \le 1-\epsilon_1^2$.
The upper bound $S_1\Delta_2 = O_n(1)$ holds immediately.
If $\epsilon_2  = o_n(1)$, then $S_1\Delta_2 = \Omega_n(1)$ follows as $\Delta_2 = \Omega_n(1)$.
Otherwise, $S_1\Delta_2 = \Omega_n(1)$ holds from $\epsilon_2 = \Omega_n(1)$.
\end{proof}

We also note the following relation between $S_1$ and $S_2$ (which can be compared to Lemma 5 in \cite{cnr20}).

\begin{fact} \label{fact:variance} 
$S_1^2/S_2 \le 1-\epsilon_1^2$.
\end{fact}

\begin{proof}
Consider a random variable $Z$ where $Z = 1/\Delta_r$ with probability 
$\epsilon_r^2/(1-\epsilon_1^2)$, for $r=2,\ldots,d$.
Then, $\Exp[Z] = S_1/(1-\epsilon_1^2)$ and $\Exp[Z^2] = S_2/(1-\epsilon_1^2)$.
Now, $\Exp[Z^2] \ge \Exp[Z]^2$ since variance is always nonnegative.
\end{proof}

Next, observe that $S_1 \ge 1-\epsilon_1^2$.
If $S_1 \le 1$, then both $S_1$ and $\sqrt{S_1}$ are constant; otherwise, $\sqrt{S_1} \le S_1$.
To summarize, it is clear that $\sqrt{S_1} \lesssim S_1 \lesssim \sqrt{S_2}$.

As for the main result in Chakraborty \etal \cite{cnr20}, their theorem requires the assumption
$\epsilon_1 \ll S_1 S_2/S_3$ and $\epsilon_1 \ll \sqrt{S_2}\Delta_2$.
In Theorem \ref{thm:spatial-search}, we replace these two assumptions with a single assumption 
$\epsilon_1 \ll \sqrt{S_1}\Delta_2$. The latter assumption is only slightly stronger
than $\epsilon_1 \ll \sqrt{S_2}\Delta_2$ since under the regime of interest, namely, $S_2/S_1^2 = \Theta_n(1)$, 
they are asymptotically equivalent.


\section{Examples}

We analyze some examples of well-known families of graphs (see Figure \ref{fig:families}).
In order to simplify the calculations, we normalize the matrices as $H/\norm{H}$ which places the eigenvalues
in $[-1,1]$ (instead of $[0,1]$). Since the two normalizations are equal up to a factor of $2$, this will not
affect our asymptotic conclusions. 

\bigskip

\begin{figure}[h]
\begin{center}
{\small
\begin{tabular}{|c||c|c|c|c||c|} \hline 
{\em Graph Family}		& {\em Groverian?}	&	$\epsilon_1$	&	$\Delta_2$	&	$S_1$	& {\em Comment} \\ \hline \hline
Clique $K_n$			& Yes				&	$n^{-1/2}$			&	1			&	1		& $\Delta_2 = 1$ \\ 
Expander				& Yes				&	$n^{-1/2}$			&	1			&	-		& $\Delta_2 = 1$ \\ 
Strongly Regular Graph	& Yes				&	$n^{-1/2}$			&	1			&	-		& $\Delta_2 = 1$ \\ 
Hamming $H(n,q)$		& Yes 				&	$q^{-n/2}$			&	$1/n$		& 	1		& $\epsilon_1 \ll \sqrt{S_1}\Delta_2$ \\ 
Johnson $J(n,k)$		& Yes				&	$n^{-k/2}$			& 	$1/k$		&	1		& $\Delta_2 = 1$ \\ 
Grassmann $G_q(n,k)$	& Yes				&	$q^{-{k(n-k)/2}}$	& 	$1-1/q$		&	-		& $\Delta_2 = 1$ \\ 
Distance Regular Graph with	& 					&						& 				&			& \\ 
classical parameters $(d,q,\alpha,\beta)$		& Yes			&	$1/\sqrt{|V(G)|}$			& 	$1-1/q + \alpha/q\beta$		&	-		& $\Delta_2 = 1$ \\ 
Cycle $C_n$				& No				&	$n^{-1/2}$			&	$n^{-2}$	&	$n$		& $\epsilon_1 \not\ll \sqrt{S_1}\Delta_2$ \\ \hline
\end{tabular}}
\vspace{.1in}
\caption{Examples of graph families and $S_1$-Groverian (or optimal spatial search) properties. 
Assume $q,k,d$ are constants and $\alpha \ll \beta$. 
The expressions involving $\epsilon_1$, $\Delta_2$, and $S_1$ are asymptotic in nature.
Some entries are missing as they do not impact the property.
}
\label{fig:families}
\end{center}
\end{figure}

\paragraph{Cliques}
The normalized eigenvalues of $K_n$ are $1$ (with multiplicity $1$) and $-1/(n-1)$ (with multiplicity $n-1$).
So, cliques are Groverian (in fact, with fidelity $1$) since $\Delta_2$ is constant. 
This was, of course, the original observation of Farhi and Gutmann \cite{fg98}.

\paragraph{Expanders}
A graph $G=(V,E)$ is called an $(n,d,c)$-expander if $G$ has $n$ vertices, maximum degree $d$, and for each set
of vertices $W$ of size $|W| \le n/2$, we have $|N(W)| \ge c|W|$, where $N(W)$ is the set of vertices not in 
$W$ but is adjacent to some vertex in $W$. Here, $c$ is called the expansion which is required to be constant. 
It is known that if $G$ is $d$-regular on $n$ vertices, then it is a $(n,d,\Delta_2/2)$-expander provided
$\Delta_2 = \Theta_n(1)$ (see \cite{as}). So, expanders are Groverian since $\Delta_2$ is constant.
This is the main observation of Chakraborty \etal \cite{cnao16}.

\paragraph{Strongly Regular Graphs}
A graph $G_n$ is called {\em strongly regular} with parameter $(n,k,a,c)$
if it is a $k$-regular on $n$ vertices where every pair of adjacent vertices have $a$ common neighbors
and every pair of non-adjacent vertices have $c$ common neighbors.
Let $\theta_1 > \theta_2$ be the non-principal eigenvalues.
Then, $k - c = \theta_1(-\theta_2)$, which implies $\theta_1 = (k-c)/(-\theta_2)$.
If $G_n$ is primitive but not a conference graph, then $-\theta_2 \ge 2$.
Therefore, $\theta_1 \le (k-c)/2 < k/2$. This shows that $\Delta_2$ is constant.
On the other hand, if $G_n$ is a conference graph then $k = (n-1)/2$ and $\theta_1 = (\sqrt{n}-1)/2$
which implies $\Delta_2 = 1-o_n(1)$. Thus, strongly regular graphs are Groverian.
This provides an alternative proof of the result due to Janmark \etal \cite{jmw14}.

\paragraph{Hamming graphs}
The Hamming graph $H(n,q)$, with $q$ is constant, has $(\ZZ/q\ZZ)^n$ as its vertices where
two $n$-tuples are connected if they differ in exactly one dimension.
As $H(n,q)$ has $q^n$ vertices with eigenvalues
$\theta_r = n(q-1) - qr$ with multiplicity $m_r = \binom{n}{r}(q-1)^r$, for $r=0,\ldots,n$,
we have $\Delta_2 = q/n(q-1) \sim 1/n$ and $\epsilon_1 = 1/q^{n/2}$. 
Since $S_1 = \Theta_n(1)$, the Hamming graph $H(n,q)$ is Groverian as
$\epsilon_1 \ll \sqrt{S_1}\Delta_2$.
To see why $S_1 = \Theta_n(1)$, notice that
\[
	S_1 = \frac{1}{q^n} \sum_{r=0}^{n} \binom{n}{r}\frac{n(q-1)}{rq}(q-1)^r
	\sim \frac{1}{q^n} \sum_{r} \binom{n+1}{r+1}(q-1)^r
	\asymp 1.
\]
For the $n$-cube or $H(n,2)$, this was observed by Childs and Goldstone \cite{cg04}.

\paragraph{Johnson graphs}
The Johnson graph $J(n,k)$ has as its vertices the set of $k$-subsets of $\{1,\ldots,n\}$, 
denoted $\binom{[n]}{k}$, where two $k$-subsets $A$ and $B$ are connected if $|A \cap B| = k-1$.
So, $J(n,k)$ has $\binom{n}{k}$ vertices with eigenvalues
$\theta_r = (k-r)(n-k-r)-r$ with multiplicity $m_r = \binom{n}{r} - \binom{n}{r-1}$, for $r=0,\ldots,k$.
Notice $\epsilon_1 = 1/\sqrt{\binom{n}{k}} \sim n^{-k/2}$ and
\[
	\Delta_2 = \frac{\theta_0 - \theta_1}{\theta_0 - \theta_k} = \frac{n}{k(n-k+1)} \sim \frac{1}{k}
\]
which is constant if $k$ is.
As $\Delta_2 = \Omega_n(1)$, this shows that $J(n,k)$, for $k \ge 3$, is Groverian. 
This recovers the results of Wong \cite{w16} and Tanaka \etal \cite{tsp}.

\begin{figure}[t]
\begin{center}
\includegraphics[width=0.4\textwidth]{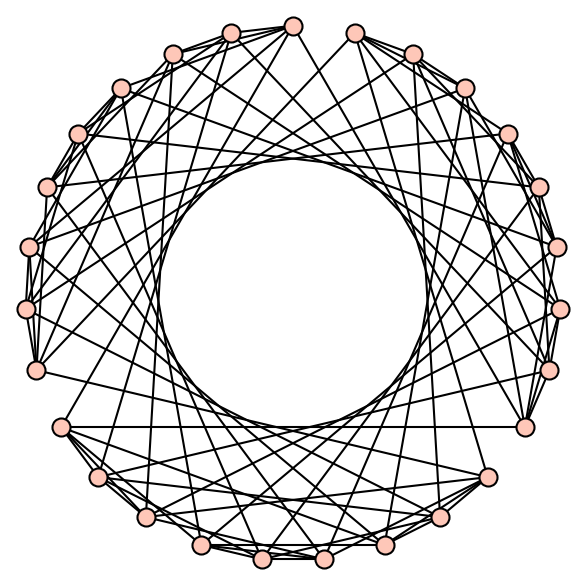}
\hspace{.5in}
\includegraphics[width=0.4\textwidth]{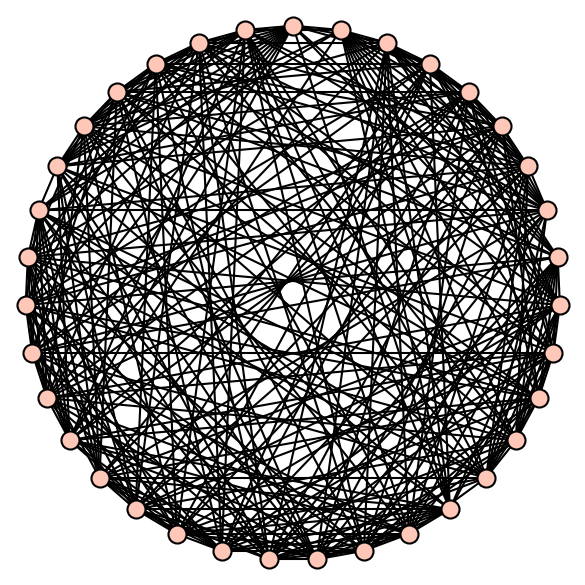}
\vspace{.1in}
\caption{Some distance-regular graphs that are Groverian: Hamming graphs $H(n,q)$ and Grassmann graphs $G_q(n,k)$,
where $q$ and $k$ are constants.
Left: Hamming graph $H(3,3)$. 
Right: Grassmann graph $G_2(4,2)$.
}
\end{center}
\end{figure}

\paragraph{Grassmann graphs}
The Grassmann graph $G_q(n,k)$ has as its vertices the set of $k$-subspaces of the vector space $\FF_q^n$
where two $k$-subspaces $A$ and $B$ are connected if $\dim(A \cap B) = k-1$.
We assume $n \ge 2k$.
The number of vertices of $G_q(n,k)$ is
\[
	N = \qbinom{n}{k} \sim q^{k(n-k)}.
\]
The eigenvalues are given by 
\[
	\theta_r = q^{r+1}\qbinom{k-r}{1}\qbinom{n-k-r}{1} - \qbinom{r}{1} \sim q^{n-r-1}
\]
with multiplicity 
\[
	m_r = \qbinom{n}{r} - \qbinom{n}{r-1}, 
\]
for $r=0,\ldots,k$.
Notice $\epsilon_1 = 1/\sqrt{N}$ and the normalized eigenvalue gap is
\[
	\Delta_2 = 1 - \frac{\theta_1}{\theta_0} \sim 1 - \frac{1}{q},
\]
which is constant if $q$ is.
As $\Delta_2 = \Omega_n(1)$, this shows that $G_q(n,k)$, for $k \ge 2$, is Groverian. 
Note $k=1$ recovers the cliques.

\paragraph{Distance-regular graphs with classical parameters}
The eigenvalues of distance-regular graphs with classical parameters $(d,q,\alpha,\beta)$
are given in Juri\v{s}i\'{c} and Vidali \cite{jv17} (see Lemma 2). In particular,
$\theta_0 = [d]_q \beta$ and $\theta_1 = [d-1]_q (\beta - \alpha) - 1$, which implies
\[
	\Delta_2 = 1 - \frac{[d-1]_q(\beta - \alpha) - 1}{[d]_q \beta} \sim 1 - \frac{1}{q} + \frac{1}{q}\frac{\alpha}{\beta}
\]
is constant provided $\alpha \ll \beta$.

\paragraph{Cycles}
The results from Section \ref{section:necessary} can be used to prove that cycles are not Groverian.
The normalized eigenvalues of $C_n$ are given by
$\theta_r = \frac{1}{2}(1 + \cos(2\pi (r-1)/n))$, $r = 1,\ldots,n$.
Using $\cos(x) \sim 1 - x^2/2$, we have $\Delta_r \sim r^2/n^2$, for small positive values of $r$.
Since $C_n$ is a circulant, $\epsilon_r \sim 1/\sqrt{n}$ for all $r=1,\ldots,n$.
Note that 
\[
	S_1 
	= \frac{1}{n} \sum_{r=1}^{n-1} \frac{1}{1 - \cos(2\pi r/n)}
	\sim 2\int_{2\pi/n}^{\pi-2\pi/n} \frac{dx}{1-\cos(x)}
	= 2\cot(\pi/n)
	= \Theta_n(n).
\]
Next, we show spatial search fails as $\gamma$ ranges over all values.
\medskip
\par\noindent{\em Case (i).} $\gamma = S_1 + \omega_n(\sqrt{n})$:
Notice Theorem \ref{thm:necessary-above-s1} applies since $\sqrt{n} = \epsilon_1 S_1 \ll \gamma - S_1 = \omega_n(\sqrt{n})$.

\ignore{
\medskip
\par\noindent{\em Case (iv).} $\gamma = S_1 - \omega_n(\sqrt{n})$:
	We argue that Theorem \ref{thm:necessary-below-s1} applies.
	The first assumption $\epsilon_1 \gamma \ll S_1 - \gamma$ holds since
	$\epsilon_1 \gamma = \sqrt{n} - \omega_n(1)$ and $S_1 - \gamma = \omega_n(\sqrt{n})$.
	To verify the other assumption with $B=3$, notice that
	\[
		\sum_{r \ge 3} \frac{|\braket{w}{y_r}|^2}{\gamma - \zeta_r} 
		\ \le \
		\frac{C}{\gamma} \int_{2\pi/n}^{\pi-2\pi/n} \frac{dx}{1 - \cos(x)}
		\ \lesssim \ \frac{n}{\gamma}
	\]
	for some constant $C$. This shows that the second assumption is satisfied as
	\[
		\epsilon_1 \sum_{r \ge 3} \frac{|\braket{w}{y_r}|^2}{\gamma - \zeta_r} \ \lesssim \ \frac{\sqrt{n}}{\gamma} = o_n(1).
	\]
}

\medskip
\par\noindent{\em Case (ii).} $\gamma = S_1 \pm o_n(n)$:
	Here, Theorem \ref{thm:necessary-around-s1} applies since $\epsilon_1^\alpha/\Delta_r > S_1 + \beta$ holds, for $\alpha \in (1,2)$.
	With $|\beta| = o_n(n)$, evidently $n^{-\alpha/2} \times n^2 \ \gg \ cn \pm o_n(n)$.

\medskip
\par\noindent{\em Case (iii).} $\gamma = O_n({n})$:
	To apply Theorem \ref{thm:necessary-around-s1}, we rewrite Equation \eqref{eqn:penultimate} as
	\begin{equation} \label{eqn:goal}
	\frac{1}{2} \sum_{\substack{r \ge 2:\\\Delta_r < \epsilon_1^\alpha/\gamma}} \frac{\epsilon_r^2}{\Delta_r}
	\ > \
	\gamma\epsilon_1^{2-\alpha} + S_1 - \gamma,
	\end{equation}
	since $\gamma > 0$.
	Because $\gamma = O_n({n})$, $S_1 = O_n(n)$, and $\epsilon_1 = 1/\sqrt{n}$, the right-hand side is at most $O_n(n)$.
	Next, we determine the set of indices $r$ so that $\Delta_r < \epsilon_1^\alpha/\gamma$ which are
	included in the summation. The smallest the upper bound $\epsilon_1^\alpha/\gamma$ can be is $1/n^{1+\alpha/2}$.
	We have
	\[
	\Delta_r \sim \frac{r^2}{n^2} \ll \frac{1}{n^{1 + \alpha/2}}, \ \ \ \mbox{ for $r = O_n(1)$.}
	\]
	Thus, if we restrict the indices for which $r \le B$, for a large enough constant $B$, we get
	\[
	\sum_{\substack{r \ge 2:\\\Delta_r < \epsilon_1^\alpha/\gamma}} \frac{\epsilon_r^2}{\Delta_r}
	\ \ge \
	\sum_{r=2}^{B} \frac{\epsilon_r^2}{\Delta_r}
	\ \sim \ 
	\sum_{r=2}^{B} \frac{n^{-1}}{r^2 n^{-2}} 
	\ = \ \Omega_n(n).
	\]
	Thus, Equation \eqref{eqn:goal} is satisfied provided $B$ is large enough; whence spatial search fails.


\begin{figure}[t]
\begin{center}
\includegraphics[width=0.4\textwidth]{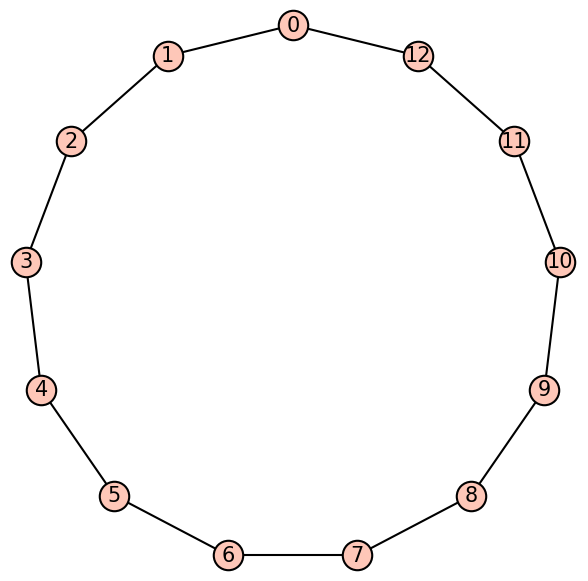}
\hspace{.5in}
\includegraphics[width=0.4\textwidth]{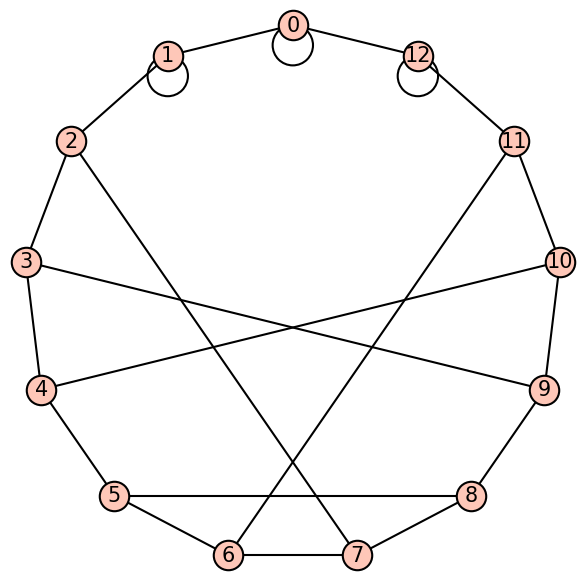}
\vspace{.1in}
\caption{Sparsity is not a good indicator: 
for prime $p$, $C_{p}$ is not Groverian, but $C_{p}$ with the extra ``matching''
$x \mapsto x^{-1}\pmod{p}$ is Groverian (as it is an expander, see Vadhan \cite{v12}).
Left: $C_{13}$. Right: the nonsimple $3$-regular expander $C_{13}$ plus the modular inverse matching.
}
\end{center}
\end{figure}


\section{Concluding remarks}

In this work, we proved a simpler characterization of graphs with the optimal spatial search property
(therein called Groverian). This improves a previous characterization obtained by Chakraborty \etal \cite{cnr20}.
We applied this characterization to recover known results about some families of Groverian graphs and also 
to find new families of Groverian graphs.
Along the way, we also proved a lower bound for spatial search on arbitrary graphs for constant fidelity.
This extends a known lower bound due to Farhi and Gutmann \cite{fg98} which holds for vertex transitive graphs
and for fidelity that is one. We also developed a family of necessary conditions for a graph to be Groverian.
Our necessary conditions, which are built upon observations developed by Chakraborty \etal \cite{cnr20}, can
be applied to provide rigorous proofs to show why some families of graphs (for example, cycles)
lack the spatial search property.

We conclude with some open questions from the present work:
\begin{enumerate}
\item Is the Groverian property determined by spectra? 
	Moreover, is the condition $\epsilon_1 \ll \sqrt{S_1}\Delta_2$ necessary for optimal characterization?
\item Can a family of graphs be Groverian with non-constant $S_1$?
	To the best of our knowledge, all families of graphs known to be Groverian satisfy $S_1 = \Theta_n(1)$.
\item When is the Groverian property (almost) periodic? 
\item How robust is the Groverian property against noise? Regev and Schiff \cite{rs08} had ruled this out
	for the discrete-time case.
\end{enumerate}


\section*{Acknowledgments}

We thank Tom Wong for helpful comments and Andris Ambainis for answering our questions about \cite{cnao16}.
C.T. and W.X. would also like to thank Warren Lord for his insightful comments.
We also would like to thank the reviewers for their constructive comments which help improve the paper.




\begin{thebibliography}{000}

\bibitem{g97}
L.~Grover (1997),
Quantum mechanics help in searching for a needle in a haystack,
{\em Physical Review Letters}, {\bf 79}:325.

\bibitem{a07}
A. Ambainis (2007), 
Quantum walk algorithm for element distinctness,
{\em SIAM J. Computing} {\bf 37}(1):210-239.

\bibitem{s04}
M. Szegedy (2004),
Quantum speed-up for Markov chain based algorithms,
{\em Proc. 45th Ann. IEEE Foundations of Computer Science}, 32-41.

\bibitem{fg98}
E.~Farhi and S.~Gutmann (1998),
Analog analogue of a digital quantum computation,
{\em Physical Review A}, {\bf 57}(4):2403.

\bibitem{cg04}
A.~Childs and J.~Goldstone (2004),
Spatial search by quantum walk,
{\em Physical Review A}, {\bf 70}:022314.

\bibitem{cnao16}
S.~Chakraborty, L.~Novo, Y.~Omar, and A.~Ambainis (2016),
Spatial search by quantum walk is optimal for almost all graphs,
{\em Physical Review Letters}, {\bf 116}:100501.

\bibitem{cnr20}
S.~Chakraborty, L.~Novo, and J.~Roland (2020),
On the optimality of spatial search by continuous-time quantum walk,
{\em Physical Review A}, {\bf 102}:032214.

\bibitem{jmw14}
J.~Janmark, D.~Meyer, and T.~Wong (2014),
Global symmetry is not necessary for fast quantum search,
{\em Physical Review Letters}, {\bf 112}:210502.

\bibitem{w16}
T.~Wong (2016),
Quantum walk search on {J}ohnson graphs,
{\em Journal of Physics A: Mathematical and Theoretical}, {\bf 49}(19):195303.

\bibitem{tsp}
H.~Tanaka, M.~Sabri, and R.~Portugal,
Spatial search on {J}ohnson graphs by continuous-time quantum walk.
arXiv:2108.01992 [math.CO].

\bibitem{hj13}
R.~Horn and C.~Johnson (2013),
{\em Matrix Analysis},
Cambridge University Press, second edition.

\bibitem{jlr}
S.~Janson, T.~\L{}uczak, and A.~Ruci\'{n}ski (2000),
{\em Random Graphs},
Wiley and Sons.

\bibitem{nc}
M.~Nielsen and I.~Chuang (2000),
{\em Quantum Computation and Quantum Information},
Cambridge University Press.

\bibitem{mw15}
D. Meyer and T. Wong (2015),
Connectivity is a Poor Indicator of Fast Quantum Search,
{\em Physical Review Letters}, {\bf 114}:110503.

\bibitem{as}
N.~Alon and J.~Spencer (2011), 
{\em The Probabilistic Method},
Wiley and Sons, 3rd edition.

\bibitem{jv17}
S.~Juri\v{s}i\'{c} and J.~Vidali (2017),
Restrictions on classical distance-regular graphs,
{\em Journal of Algebraic Combinatorics}, {\bf 46}:571--588.

\bibitem{v12}
S. Vadhan (2012), 
{\em Pseudorandomness},
Foundations and Trends in Theoretical Computer Science, Vol. 7, Nos. 1-3, 1-336.

\bibitem{rs08}
O. Regev and L. Schiff (2008),
Impossibility of a Quantum Speed-Up with a Faulty Oracle,
{\em Proc. 35th International Colloquium on Automata, Languages, and Programming}, 773-781.

\end{thebibliography}
\end{document}